\documentclass[review]{elsarticle}

\usepackage{lineno,hyperref}
\modulolinenumbers[5]

\usepackage[T1]{fontenc}
\usepackage[latin9]{inputenc}
\usepackage{color}
\usepackage{array}
\usepackage{float}
\usepackage{mathrsfs}
\usepackage{mathtools}
\usepackage{amsthm}
\usepackage{amstext}
\usepackage{amssymb}
\usepackage{stmaryrd}
\usepackage{graphicx}
\usepackage{pgfplots}

\allowdisplaybreaks

\makeatletter
\numberwithin{equation}{section}  

\theoremstyle{plain}
\newtheorem{theorem}{\protect\theoremname}
\theoremstyle{plain}

\theoremstyle{plain}
\newtheorem{corollary}{\protect\corollaryname}
\theoremstyle{plain}
\newtheorem{lemma}{\protect\lemmaname}
\theoremstyle{definition}

\theoremstyle{definition}
\newtheorem{definition}{\protect\definitionname}
\theoremstyle{definition}
\newtheorem{remark}{\protect\remarkname}
\theoremstyle{definition}

\AtBeginDocument{
	
}

\makeatother

\providecommand{\corollaryname}{Corollary}
\providecommand{\examplename}{Example}
\providecommand{\lemmaname}{Lemma}
\providecommand{\propositionname}{Proposition}
\providecommand{\theoremname}{Theorem}
\providecommand{\definitionname}{Definition}
\providecommand{\remarkname}{Remark}
\providecommand{\acknowledgementsname}{Acknowledgements}

\journal{Finite Fields and Their Applications}

\begin{document}
	
	\begin{frontmatter}
		
		\title{Construction of three class of at most four-weight binary linear codes 
			and their applications} 
		
		
		\author[mymainaddress]{Tonghui Zhang} 
		\ead{zhangthvvs@126.com}

		\address[mymainaddress]{School of Mathematics and Statistics, Fujian Normal University,\\Fuzhou, Fujian, 350117, P. R. China}

		\author[2]{Pinhui Ke\corref{mycorrespondingauthor}}
		\cortext[mycorrespondingauthor]{Corresponding author}
		
		\ead{keph@fjnu.edu.cn}
		
		\address[2]{Key Laboratory of Analytical Mathematics and Applications (Ministry of Education),\\ Fujian Normal University, Fuzhou, Fujian, 350117, P. R. China}
		\author[3]{Zuling Chang} 
		\ead{zuling\_chang@zzu.edu.cn}
		
		\address[3]{School of Mathematics and Statistics, Zhengzhou University, \\Zhengzhou, Henan, 450001, P. R. China}

		\begin{abstract}
			
			Three classes of binary linear codes with at most four nonzero weights were constructed in this paper, in which two of them are projective three-weight codes. As applications, $s$-sum sets for any odd $ s > 1$ were constructed.
		\end{abstract}
		
		\begin{keyword}
			projective code, defining set, weight distribution, $s$-sum set
			\MSC[2010] 94B15 \sep  11T71
		\end{keyword}
		
	\end{frontmatter}
	
	\vspace{2em}
	\noindent\textbf{Note.} This manuscript has been submitted for publication.
	\section{Introduction}
	
	Denote $\mathbb{F}_{p^m}$ as the finite field with $p^m$ elements and $\mathbb{F}_{p^m}^*=\mathbb{F}_{p^m}\setminus \{0\} $, where $p$ is a prime. An [$n,k,d$] linear code $C$ of length $n$ over $\mathbb{F}_{q}$ is a $k$-dimensional linear subspace of $\mathbb{F}^n_{q}$
	with minimum (Hamming) distance $d$.
	We refer to an $[n,k,d]$ code as optimal if no $[n,k,d+1]$ code exists. Conversely, an $[n,k,d-1]$ code is almost optimal if the $[n,k,d]$ code is already optimal.
	Let $A_i$ denote the number of codewords with weight $i$ in a code of length $n$. The sequence $\left(1,A_1,A_2,\dots,A_n \right)$ is called the weight distribution of $C$ and the polynomial $1+A_1z+A_2z^2+\cdots+A_nz^n$ is called the weight enumerator of $C$. 
	The code $C$ is called $t$-weight code if the number of nonzero $A_j$ in the sequence $(A_1,A_2,\dots, A_n ) $ is equal to $t$.
	The weight distribution of a linear code reflects the information of the capabilities of error detection and correction but also allows the computing of the error probability of error detection and correction.
	The dual code $C^{\perp}$ of $C$ is defined by $C^{\perp}=\{\textbf{x}\in \mathbb{ F}_q^n:\textbf{x} \cdot \textbf{c}=0 \textup{ for all } \textbf{c} \in C \}$.
	A linear code whose dual  has the minimal distance $d^{\perp} \geq 3$ is called a projective code.
	
	Linear codes get much attention due to their applications in secret sharing schemes, strongly regular graphs, association schemes and authentication codes $\textit{etc}$. 
	In particular, projective two-weight and three-weight codes have received a lot of attention due to their close connections with finite projective spaces, strongly regular graphs and $s$-sum sets.

	Projective linear codes were constructed by using almost difference sets and two to one functions in  \cite{heng2023almostdifference} and \cite{Mes2023two to one}, respectively.
	Moreover, by using the second generic construction described in \cite{li2020}, several new classes of projective  two-weight and three-weight codes were presented in \cite{cheng2023}, \cite{heng2021}, \cite{liu2023}, \cite{zhu2023twoweight} and \cite{zhu2023 $s$-sum}.
	Also some Lee-weight projective codes were constructed in \cite{shi2017}, \cite{shi2014} and \cite{tang2023}.

	Let $\textup{Tr}_{p^m/p}$ denote the trace function from $\mathbb{F}_{p^m}$ to $\mathbb{F}_p$. 
	For a set $D=\{d_1,d_2,\dots,d_n\} \subseteq \mathbb{F}_{p^m}$, Ding and Niederreiter in \cite{ding2007} defined a generic class of linear codes of length $n=\#D$ over $\mathbb{F}_p$ as
	\begin{equation}
		C_D=\left\{\textbf{c}(a)=\left(\textup{Tr}_{p^m/p}(ad_1),\textup{Tr}_{p^m/p}(ad_2),\dots,\textup{Tr}_{p^m/p}(ad_n)\right):a\in\mathbb{F}_{p^m} \right \}.\notag
	\end{equation}
	Here $D$ is called the defining set of $C_D$.
	
	For a given positive integer $d$ and an odd prime  $p$, Zhu $ \textit{et}$ $ \textit{al}$ \cite{zhu2023twoweight} introduced the defining set 
	\begin{align*}
		D_0=\{(x,y) \in \mathbb{ F}_{p^m}^* \times  \mathbb{ F}_{p^m}: \textup{Tr}_{p^m/p}(yx^{d+1})=0 \},
	\end{align*}
	and defined the code 
	\begin{align*}
		C_{D_0}=\{\left(\textup{Tr}_{p^m/p}(ayx^d+bx)\right)_{(x,y) \in D_0} :(a,b) \in \mathbb{ F}_{p^m} \times  \mathbb{ F}_{p^m} \}.
	\end{align*} 
	Two new classes of projective two-weight codes
	were obtained and then two new classes of strongly regular graphs were given in \cite{zhu2023twoweight}. Inspired by this construction, we consider the linear codes
	\begin{equation}\label{code}
		C_{D_i}=\left \{\textbf{c}(a,b)=\left(\textup{Tr}(axy+bx)\right)_{(x,y)\in {D_i}}: a,b \in \mathbb{F}_{q^m} \right \},
	\end{equation}
where $i=1,2,3$ and the defining sets
	\begin{equation}\label{det1}
		D _1=\left\{(x,y)\in \mathbb{F}_{2^m}^* \times \mathbb{F}_{2^m} :\textup{Tr}\left(yx^2+y\right)=0\right\},
	\end{equation}
	\begin{equation}\label{det2}
		D_2=\left\{(x,y)\in \mathbb{F}_{2^m}^* \times \mathbb{F}_{2^m} :\textup{Tr}\left(yx^2+x+y\right)=0\right\},
	\end{equation}
	
	\begin{equation}\label{det3}
		D_3=\left\{(x,y)\in \mathbb{F}_{2^m}^* \times \mathbb{F}_{2^m} :\textup{Tr}\left(yx^2+xy\right)=0\right\}.
	\end{equation}
	When $m$ is even, $C_{D_2}$ and $C_{D_1}$ have the same weight distribution, therefore we only investigate
	the weight distribution of $C_{D_2}$ when $m$ is odd. 
	Our results indicate that three classes of linear codes with at most four weights and two of them are projective three-weight codes which can be used to construct $s$-sum sets.
	
	The rest of this paper is arranged as follows.
	Section 2 provides a brief summary of relevant properties of trace functions and additive characters.
	Section 3 presents some auxiliary results that will be utilized in the subsequent sections.
	The main results are outlined in Section 4, where three classes of projective codes are scrutinized.
	Section 5 demonstrates the application of these codes in $s$-sum sets construction.
	Section 6 concludes the article.

	\section{Preliminaries}
	In this section, we introduce some basic concepts of trace functions and  additive characters on finite fields. 
		
		

		\begin{definition}(\cite[Definition 2.22]{lidl1997})
			For $\alpha \in \mathbb{ F}_{2^m}$, the trace function $\textup{Tr}(\alpha) $ of  $\alpha$ over $\mathbb{ F}_{2^m}$ is defined by
			\begin{align*}
				\textup{Tr}(\alpha)=\alpha+\alpha^2+\cdots +\alpha^{2^{m-1}}.
			\end{align*} 
		\end{definition}
		
		\begin{remark}(\cite[Theorem 2.23]{lidl1997})
			This paper mainly uses the following two properties of the trace function.	
			\begin{enumerate}[(i)]
				\item $\textup{Tr}(\alpha+\beta)=\textup{Tr}(\alpha)+\textup{Tr}(\beta)$ for all $\alpha,\beta \in \mathbb{ F}_{2^m}$,
				\item  When $m$ is odd, $\textup{Tr}(1)=1$.
			\end{enumerate}
		\end{remark}
		
		Let $\zeta_p$ be a primitive $p$-th root of unity. When $p=2$, $\zeta_2=-1$. An additive character $\chi_i$ over $\mathbb{ F}_{2^m}$ to  $\mathbb{ F}_{2}$ is a homomorphism from $\mathbb{ F}_{2^m}$ into the multiplicative group composed by $\zeta_2$ of the complex number of absolute value $1$. For any $i\in\mathbb{F}_{2^m}$, $\chi_i(x)=(-1)^{\textup{Tr}(ix)}$, where $x\in\mathbb{ F}_{2^m}$. If $i=0$, $\chi_0(x)=1 $ for all $x\in\mathbb{ F}_{2^m}$ and $\chi_0$ is the trivial additive character of $\mathbb{F}_{2^m}$. If $i=1$, $\chi_1$ is called the canonical additive character of $\mathbb{F}_{2^m}$, which is recorded as $\chi$ in this paper. The orthogonality of the canonical additive character is given by
		\begin{equation}
			\sum\limits_{x \in \mathbb{F}_{2^m}}\chi(x)=\sum\limits_{x \in \mathbb{F}_{2^m}}(-1)^{\textup{Tr}(x)}=0.
			\notag
		\end{equation}

		
		In order to determine the multiplicity of each weight of the desired codes, we need the Pless power moments for linear codes. 
		
		\begin{lemma}\textup{(\cite[p.259, the Pless power moments]{huffman2010})}\label{Pless}
			Let $C$ be an $[n,k]$ code over $\mathbb{F}_q$ with weight distribution $(1,A_1,\dots,A_n)$ and $C^{\perp}$ be its dual code with weight distribution $(1,A^{\perp}_1,\dots,A^{\perp}_n)$, the first three Pless power moments are as follows,
			\begin{align*}
				&\sum^n_{j=0}A_j=q^k,\\
				&\sum^n_{j=0}jA_j=q^{k-1}(qn-n-A^{\perp}_1),\\
				&\sum^n_{j=0}j^2A_j=q^{k-2}\big((q-1)n(qn-n+1)-(2qn-q-2n+2)A_1^{\perp}+2A_2^{\perp}  \big).
			\end{align*}
		\end{lemma}
		
		\section{Some auxiliary results}
		In this section, we  give some important auxiliary results which will be used in the sequel.
		Obviously, the codeword is the zero codeword when $(a,b)=(0,0)$. Hence we will default $a$ and $b$ not to be zero at the same time in the following.

		\begin{lemma} \label{a}
			For $a \in \mathbb{F}_{2^m}^*$, suppose that $f(x)=x^2+ax+1 \in \mathbb{F}_{2^m}[x]$ is a function from  $\mathbb{ F}_{2^m}$ to $\mathbb{ F}_{2^m}$, then 
			$f(x)=0$ has two nonzero solutions if and only if $a =r+r^{-1}$, where $r \in \mathbb{F}_{2^m}^* \setminus \{1\}$ and $r^{-1}$ represents the inverse of $r$.
		\end{lemma}
		
		\begin{proof}
			Let  $x_1$, $x_2 \in \mathbb{F}_{2^m}^* \setminus \{1\} $ be the nonzero solutions of $f(x)=0$, and according to Vieta's Theorem, $x_1+x_2=a$, $x_1x_2=1$, so that $a=x_1+x_1^{-1}$.
			Conversely, if $a =r+r^{-1}$,  then obviously  $r$ and $r^{-1}$ are two nonzero solutions of $f(x)=0$. 
		\end{proof}
		
		\begin{corollary}\label{a1}
			Let $\Upsilon_1$ be the set of all $a$ that cause $f(x)$ in Lemma \ref{a}
			to have nonzero solutions, then $\Upsilon_1$ can also be expressed as
			\begin{equation}
				\Upsilon_1=\{a=r+r^{-1},a \in \mathbb{F}_{2^m}^* \}.\notag
			\end{equation}
			Simple verification shows that $\big|\Upsilon_1\big|=2^{m-1}-1$.
		\end{corollary}
		
		\begin{corollary}\label{a2}
			Suppose $g(x)=x^2+(a+1)x $ to be the function from $\mathbb{ F}_{2^m}$ to $\mathbb{ F}_{2^m}$, for $a \in \mathbb{F}_{2^m}^*$, the following results hold.
			\begin{enumerate}[(i)]
				\item  $g(x)=0$ has a nonzero solution if and only if $a \ne 1 $.
				\item  	Define that $\Upsilon_2$ be the set of all $a$ that cause $g(x)$ to have nonzero solutions, then 
				$ \Upsilon_2=\{a \in \mathbb{F}_{2^m}^*\setminus\{1\} \}	$, obviously, $|\Upsilon_2|=2^{m}-2$.
			\end{enumerate}	
		\end{corollary}
		
		All proofs of lemmas below are mainly concerned with the orthogonality of  additive characters.
		
		\begin{lemma}\label{T}
			Let $\mathbb{E} $ be a subset of $\mathbb{F}_{2^m}^*$ and $T_i=\big|\{(e,b) \in \mathbb{E} \times \mathbb{F}_{2^m}^*:\textup{Tr}(eb)=i\}\big|$, $i=0,1$.  It is easy to know that
			\begin{align*}
				T_i=
				\begin{cases}
					\big|\mathbb{E} \big|(2^{m-1}-1) ,&  \textup{if } i=0,\\
					\big|\mathbb{E} \big|2^{m-1} ,&  \textup{if } i=1.
				\end{cases}
			\end{align*}	
		\end{lemma}
		\begin{proof}
			\begin{align*}
				T_i
				&=\dfrac{1}{2}\sum\limits_{z\in \mathbb{F}_2}\sum\limits_{e\in \mathbb{E}}\sum\limits_{b\in \mathbb{F}_{2^m}^*}(-1)^{ \left( \textup{Tr}(eb)-i\right) z}\\
				&=\dfrac{1}{2}|\mathbb{E}|(2^m-1)+\dfrac{1}{2}(-1)^{-i}\sum\limits_{e\in \mathbb{E}}\sum\limits_{b\in \mathbb{F}_{2^m}^*}(-1)^{\textup{Tr}(eb)}\\
				&=\dfrac{1}{2}|\mathbb{E}|(2^m-1)-\dfrac{1}{2}(-1)^{-i}|\mathbb{E}|\\
				&=\begin{cases}
					\big|\mathbb{E} \big|(2^{m-1}-1) ,&  \textup{if } i=0,\\
					\big|\mathbb{E} \big|2^{m-1} ,&  \textup{if } i=1.
				\end{cases}
			\end{align*}
		\end{proof}

		\begin{lemma}\label{S1}
			Define 
			$S_1=\sum\limits_{x\in \mathbb{F}_{2^m}^*}\sum\limits_{y\in \mathbb{F}_{2^m}}(-1)^{ \textup{Tr}(axy+bx)}$, then
			\begin{align*}
				S_1=\begin{cases}
					-2^m ,& \textup{if } a=0,b \ne 0,\\
					0    ,&  \textup{otherwise}.
				\end{cases}
			\end{align*}
		\end{lemma}
		\begin{proof}
			\begin{enumerate}[(i)]
				\item When $a=0$, $b \ne 0$,
				\begin{align*}
					S_1&=2^m\sum\limits_{x\in \mathbb{F}_{2^m}^*}(-1)^{ \textup{Tr}(bx)} = -2^m.
				\end{align*}
				\item When $a \ne 0$, $b \in \mathbb{F}_{2^m} $,
				\begin{align*}
					S_1&=\sum\limits_{x\in \mathbb{F}_{2^m}^*}(-1)^{ \textup{Tr}(bx)}\sum\limits_{y\in \mathbb{F}_{2^m}}(-1)^{ \textup{Tr}(axy)}=0.
				\end{align*}
			\end{enumerate}
			The proof is completed.	
		\end{proof}
		
		\begin{lemma}\label{S2}
			Define	$S_2=\sum\limits_{x\in \mathbb{F}_{2^m}^*}\sum\limits_{y\in \mathbb{F}_{2^m}}(-1)^{\textup{Tr}\left(yx^2+y\right)+ \textup{Tr}(axy+bx)  }  $
			then	
			\begin{align*}
				S_2=\begin{cases}
					-2^m ,&  \textup{if }  a=0,  b \in \mathbb{F}_{2^m}^*, \textup{Tr}(b)=1,\\
					0 	,& \textup{if }   a \notin \Upsilon_1  \cup \{0\}, b \in \mathbb{F}_{2^m} \\
					&	\textup{or } a \in \Upsilon_1, b \in \mathbb{F}_{2^m}^*, \textup{Tr}\left( ab\right) =1,\\
					2^m ,& \textup{if }  a=0 , b \in \mathbb{F}_{2^m}^*,\textup{Tr}(b)=0,\\
					2^{m+1} ,& \textup{if }  a \in \Upsilon_1, b =0,\\
					-2^{m+1} \textup{ or } 2^{m+1} ,&  \textup{if }  a \in \Upsilon_1, b \in \mathbb{F}_{2^m}^*, \textup{Tr}\left( ab\right) =0.
				\end{cases}
			\end{align*}
			
		\end{lemma}

		\begin{proof}
			\begin{enumerate}[(i)]
				\item When $a =0 $, $b \ne 0$,
				\begin{align*}
					S_2&=\sum\limits_{x\in \mathbb{F}_{2^m}^*}(-1)^{\textup{Tr}(bx)}\sum\limits_{y\in \mathbb{F}_{2^m}}(-1)^{\textup{Tr}\left((x^2+1)y\right)  }\\
						&=(-1)^{\textup{Tr}(b)}2^m
						=\begin{cases}
							2^m , & \textup{if } \textup{Tr}(b)=0,\\
							-2^m , & \textup{if } \textup{Tr}(b)=1.
						\end{cases}	
					\end{align*}	
					
					\item When $a \ne 0 $, $b = 0$,
					\begin{align*}
						S_2
						&=\sum\limits_{x\in \mathbb{F}_{2^m}^*}\sum\limits_{y\in \mathbb{F}_{2^m}}(-1)^{\textup{Tr}\left(yx^2+y\right)+ \textup{Tr}(axy)  } \\
						&=\sum\limits_{x\in\mathbb{F}_{2^m}^*}\sum\limits_{y\in \mathbb{F}_{2^m}}(-1)^{\textup{Tr}\left((x^2+ax+1)y\right)  }\\
						&=\begin{cases}
							2^{m+1} , &  \textup{if } a \in \Upsilon_1,\\
							0       , &  \textup{if } a \notin \Upsilon_1.
						\end{cases}
					\end{align*}
					
					\item\label{same} When $a  \ne 0 $, $b \ne 0$,
					
					\begin{align*}
						S_2&=	\sum\limits_{x\in \mathbb{F}_{2^m}^*}\sum\limits_{y\in \mathbb{F}_{2^m}}(-1)^{\textup{Tr}\left(yx^2+y\right)+ \textup{Tr}(axy+bx)  }\\
						&=	\sum\limits_{x\in \mathbb{F}_{2^m}^*}(-1)^{\textup{Tr}\left(bx\right)}\sum\limits_{y\in \mathbb{F}_{2^m}}(-1)^{\textup{Tr}\left((x^2+ax+1)y\right) }\\
						&=\begin{cases}
							\left( (-1)^{\textup{Tr}\left(bx_1\right)}+ (-1)^{\textup{Tr}\left(bx_2\right)}\right) 2^m , & \textup{if } a \in \Upsilon_1,\\
							0  ,& \textup{if } a \notin \Upsilon_1,
						\end{cases}	
					\end{align*}	
					where $x_1$ and $x_2$ are two nonzero solutions of $x^2+ax+1=0$, according to Vieta's Theorem, $x_1+x_2=a$. And from that we can derive 
					\begin{align*}
						(-1)^{\textup{Tr}\left(bx_1\right)}+ (-1)^{\textup{Tr}\left(bx_2\right)}
						&=
						\begin{cases}
							2 \textup{ or } -2 , & \textup{if } (-1)^{\textup{Tr}\left(ab\right)}=1,\\
							0 , &  \textup{if } (-1)^{\textup{Tr}\left(ab\right)}=-1,
						\end{cases}\\
						&=
						\begin{cases}
							2 \textup{ or } -2 , &  \textup{if } \textup{Tr}\left( ab\right) =0,\\
							0  , &   \textup{if } \textup{Tr}\left( ab\right) =1.
						\end{cases}
					\end{align*}
					
					Then we have		
					\begin{align*}
						S_2=\begin{cases}
							2^{m+1} \textup{ or } -2^{m+1}  , & \textup{if } a \in \Upsilon_1, \textup{Tr}\left( ab\right) =0,\\
							0  , & \textup{if }a \notin \Upsilon_1 \textup{ or }  a \in \Upsilon_1, \textup{Tr}\left( ab\right) =1.
						\end{cases}
					\end{align*}	
				\end{enumerate}
				The result can be obtained after sorting.
				
			\end{proof}
			
			\begin{lemma}\label{S3}
				When $m$ is odd, define	
				\begin{align*}
					S_3=\sum\limits_{x\in \mathbb{F}_{2^m}^*}\sum\limits_{y\in \mathbb{F}_{2^m}}(-1)^{\textup{Tr}\left(yx^2+x+y\right)+ \textup{Tr}(axy+bx)  },  
				\end{align*}
				then we have
				\begin{align*}
					S_3=\begin{cases}
						-2^m ,& \textup{if }  a=0, b \in \mathbb{F}_{2^m}^*, \textup{Tr}(b)=0,\\
						2^m  ,& \textup{if }  a=0, b \in \mathbb{F}_{2^m}^*, \textup{Tr}(b)=1,\\
						0    ,& \textup{if }  a \notin \Upsilon_1  \cup \{0\} ,  b \in \mathbb{F}_{2^m} \textup{ or } a \in \Upsilon_1,  b \in \mathbb{F}_{2^m}, \textup{Tr}((b+1)a)=1,\\
						-2^{m+1} \textup{ or } 2^{m+1}  ,& \textup{otherwise}.
					\end{cases}
				\end{align*}
				
			\end{lemma}
			
			\begin{proof}
				\begin{enumerate}[(i)]
					\item When $a=0$, $b \ne 0$,
					\begin{align*}
						S_3&=\sum\limits_{x\in \mathbb{F}_{2^m}^*}(-1)^{\textup{Tr}((b+1)x)}\sum\limits_{y\in \mathbb{F}_{2^m}}(-1)^{\textup{Tr}\left((x^2+1)y\right)}\\
						&=-(-1)^{\textup{Tr}(b)}2^m 
						=\begin{cases}
							2^m  ,& \textup{if }   \textup{Tr}(b)=1,\\
							-2^m  , &  \textup{if }  \textup{Tr}(b)=0.
						\end{cases}
					\end{align*}
					\item When $a \ne 0$, $b \in \mathbb{F}_{2^m}$,
					\begin{align*}
						S_3&=\sum\limits_{x\in \mathbb{F}_{2^m}^*}(-1)^{\textup{Tr}((b+1)x)}\sum\limits_{y\in \mathbb{F}_{2^m}}(-1)^{\textup{Tr}\left((x^2+ax+1)y\right)}\\
						&=\begin{cases}
							\left( (-1)^{\textup{Tr}\left((b+1)x_1\right)}+ (-1)^{\textup{Tr}\left((b+1)x_2\right)}\right) 2^m ,& \textup{if } a \in \Upsilon_1,\\
							0 ,& \textup{if } a \notin \Upsilon_1,
						\end{cases}	
					\end{align*}
					where $x_1$ and $x_2$ are two nonzero solutions of $x^2+ax+1=0$ in $\mathbb{ F}_{2^m}$. Similar to the proof of the case \eqref{same} in Lemma \ref{S2}, we have 
					\begin{align*}
						S_3=\begin{cases}
							2^{m+1} \textup{ or } -2^{m+1} ,& \textup{if } a \in \Upsilon_1, \textup{Tr}\left( a(b+1)\right) =0,\\
							0  ,& \textup{if }a \notin \Upsilon_1 \textup{ or }  a \in \Upsilon_1, \textup{Tr}\left( a(b+1)\right) =1.
						\end{cases}
					\end{align*}		
				\end{enumerate}
				This completes the proof of this lemma.
			\end{proof}

			\begin{lemma}\label{S4}
				Define	$S_4=\sum\limits_{x\in \mathbb{F}_{2^m}^*}\sum\limits_{y\in \mathbb{F}_{2^m}}(-1)^{\textup{Tr}\left(yx^2+xy\right)+ \textup{Tr}(axy+bx)  }  $
				, then	
				\begin{align*}
					S_4=\begin{cases}
						0 ,&	 \textup{if }  a=1, b \in \mathbb{F}_{2^m}  ,\\
						2^m ,&  \textup{if }  a=0,  b \in \mathbb{F}_{2^m}^* , \textup{Tr}(b)=0   ,\\
						& \textup{or }  a \in \Upsilon_2 ,  b \in \mathbb{F}_{2^m} , \textup{Tr}((a+1)b)=0,\\
						-2^m ,&\textup{if }   a=0,  b \in \mathbb{F}_{2^m}^* , \textup{Tr}(b)=1, \\
						&  \textup{or } a  \in \Upsilon_2,   b \in \mathbb{F}_{2^m}^* , \textup{Tr}((a+1)b)=1   . 
					\end{cases}
				\end{align*}
				
			\end{lemma}

			\begin{proof}
				\begin{enumerate}[(i)]
					\item  When $a  = 0 $, $b \ne 0$,
					\begin{align*}
						S_4&=\sum\limits_{x\in \mathbb{F}_{2^m}^*}(-1)^{\textup{Tr}(bx)}\sum\limits_{y\in \mathbb{F}_{2^m}}(-1)^{\textup{Tr}\left((x^2+x)y\right)  }\\
						&=(-1)^ {\textup{Tr}(b)} 2^m\\
						&=\begin{cases}
							2^m ,&  \textup{if } \textup{Tr}(b)=0,\\
							-2^m ,&  \textup{if } \textup{Tr}(b)=1.
						\end{cases}	
					\end{align*}			
					\item  When $a  \ne 0 $, $b = 0$,
					\begin{align*}	
						S_4&=\sum\limits_{x\in \mathbb{F}_{2^m}^*}\sum\limits_{y\in \mathbb{F}_{2^m}}(-1)^{\textup{Tr}\left((x^2+(a+1)x)y\right)  }\\
						&=\begin{cases}
							0 ,& \textup{if } a=1,\\
							2^m ,& \textup{otherwise}.
						\end{cases}		
					\end{align*}
					
					\item When $a \ne 0$, $b \ne 0$,			
					\begin{align*}
						S_4&=\sum\limits_{x\in \mathbb{F}_{2^m}^*}(-1)^{\textup{Tr}(bx)} \sum\limits_{y\in \mathbb{F}_{2^m}}(-1)^{\textup{Tr}\left((x^2+(a+1)x)y\right)}\\
						&=\begin{cases}
							0 ,&\textup{if } a=1,\\
							(-1)^ {\textup{Tr}((a+1)b)} 2^m, & \textup{if } a\ne 1,
						\end{cases}\\
						&=\begin{cases}
							0 ,&\textup{if } a=1,\\
							2^m ,& \textup{if } a\ne 1,\textup{Tr}((a+1)b)=0,	\\
							-2^m ,& \textup{if } a\ne 1,\textup{Tr}((a+1)b)=1.
						\end{cases}
					\end{align*}
				\end{enumerate}
				The proof is completed.	
			\end{proof}

			\section{ Main results and proofs }
			
			This section gives the results of our constructions including parameters, weight distributions, proofs and some concrete examples.
			\subsection{Main results}
			
			\begin{theorem}\label{thm1}
				Suppose that $C_{D_1}$ and $D_1$ are defined by \eqref{code} and \eqref{det1}, respectively, then $C_{D_1}$ is a $[2^{2m-1},2m,2^{m-1}(2^{m-1}-1)]$ projective three-weight linear code with weight distribution given in Table \ref{table1}.
			\end{theorem}

			\begin{table}[H]
				\caption[short text]{The weight distribution of $C_{D_1}$ }
				\centering
				\label{table1}
				\begin{tabular}{cc}
					\hline
					weight & multiplicity \\ \hline
					0 & 1 \\	
					$ 2^{m-1}(2^{m-1}-1)$ & $2^{m-2}(2^{m-1}-1)$\\
					$ 2^{2m-2}$ & $3 \cdot 2^{2m-2}-1$ \\
					$2^{m-1}(2^{m-1}+1)$ & $2^{m-2}(2^{m-1}+1)$\\\hline
				\end{tabular}
			\end{table}
			
			\begin{theorem} \label{Thm2}
				When $m$ is odd, suppose that $C_{D_2}$ and $D_2$ are  defined by \eqref{code} and \eqref{det2}, respectively, then $C_{D_2}$ is a $[2^{m}(2^{m-1}-1),2m,2^{m}(2^{m-2}-1)]$ projective three-weight linear code with weight distribution given in Table \ref{table2}.	
				
			\end{theorem}
			\begin{table}[H]
				\caption[short text]{The weight distribution of $C_{D_2}$ }
				\centering
				\label{table2}
				\begin{tabular}{cc}
					\hline
					weight & multiplicity \\ \hline
					0 & 1 \\
					$2^{m}(2^{m-2}-1)$ & $ 2^{m-2}(2^{m-1}-1) $\\	
					$ 2^{m-1}(2^{m-1}-1)$ & $3 \cdot 2^{2m-2}$\\
					$ 2^{2m-2}$ & $2^{m-2}(2^{m-1}+1)-1$ \\\hline
				\end{tabular}
			\end{table}

			\begin{theorem}\label{thm3}
				Suppose that $C_{D_3}$ and $D_3$ are defined by \eqref{code} and  \eqref{det3}, respectively, then $C_{D_3}$ is a $[2^{2m-1},2m,2^{2m-2}(2^m-1)]$ projective four-weight linear code with weight distribution given in Table \ref{table3}.
			\end{theorem}

			\begin{table}[H]
				\caption[short text]{The weight distribution of $C_{D_3}$ }
				\centering
				\label{table3}
				\begin{tabular}{cc}
					\hline
					weight & multiplicity \\ \hline
					0 & 1 \\	
					$ 2^{m-2}(2^{m}-1)$ & $2^{m}(2^{m-1}-1)$\\
					$ 2^{2m-2}$ & $2^m+2^{m-1}-1$ \\
					$2^{m-2}(2^{m}+1)$ &  $2^m(2^{m-1}-1)$ \\
					$2^{m-1}(2^{m-1}+1)$ & $2^{m-1}$\\\hline
				\end{tabular}
			\end{table}
			
			\subsection{The proofs of main results }
			\begin{enumerate}[\textbf{The proof for Theorem \ref{thm1}}]
				\item 
				
			\end{enumerate}
			
			According to the definition, the length of $C_{D_1}$ equals to 
			\begin{align*}
				n_1&=|D_1|\\
				&=\dfrac{1}{2}\sum\limits_{z \in \mathbb{F}_2}\sum\limits_{x \in \mathbb{F}_{2^m}^*}\sum\limits_{y \in \mathbb{F}_{2^m}}(-1)^{z\textup{Tr}\left((x^2+1)y\right)}\\
				&=(2^m-1)2^{m-1}+\dfrac{1}{2}\sum\limits_{x \in \mathbb{F}_{2^m}^*}\sum\limits_{y \in \mathbb{F}_{2^m}}(-1)^{\textup{Tr}\left((x^2+1)y\right)}\\
				&= (2^m-1)2^{m-1}+2^{m-1}\\
				&=2^{2m-1}.
			\end{align*}
			
			Define 
			\begin{align*}
				N_1=\big|\left\{(x,y)\in \mathbb{F}_{2^m}^* \times \mathbb{F}_{2^m}:\textup{Tr}\left(yx^2+y\right)=0 \textup{ and } \textup{Tr}(axy+bx)=0 \right\}\big|,
			\end{align*}
			then the weights of codeword $\textbf{c}$ from linear code $C_{D_1}$, 
			\begin{align*}
				wt_1(\textbf{c})&=n_1-N_1\\
				&=n_1-\dfrac{1}{2^2}\sum\limits_{z_1\in \mathbb{F}_2}\sum\limits_{z_2\in \mathbb{F}_2}\sum\limits_{x\in \mathbb{F}_{2^m}^*}\sum\limits_{y\in \mathbb{F}_{2^m}}(-1)^{z_1\textup{Tr}\left(yx^2+y\right)+ z_2\textup{Tr}(axy+bx)  }\\
				&=
				n_1-\dfrac{1}{2}n_1-
				\dfrac{1}{2^2}\sum\limits_{x\in \mathbb{F}_{2^m}^*}\sum\limits_{y\in \mathbb{F}_{2^m}}(-1)^{ \textup{Tr}(axy+bx)}
				-\dfrac{1}{2^2}\sum\limits_{x\in \mathbb{F}_{2^m}^*}\sum\limits_{y\in \mathbb{F}_{2^m}}(-1)^{\textup{Tr}\left(yx^2+y\right)+ \textup{Tr}(axy+bx)  } \\
				&=2^{2m-2}-\dfrac{1}{2^2}S_1-\dfrac{1}{2^2}S_2.
			\end{align*}
			Plugging in Lemma \ref{S1} and Lemma \ref{S2}, we get 
			\begin{align*}
				wt_1(\textbf{c})=
				\begin{cases}
					2^{2m-2} ,& \textup{if } a=0, b \in \mathbb{F}_{2^m}^* \textup{Tr}(b)=0\\
					&	\textup{or }   a \notin \Upsilon_1 \cup \{0\}, b \in \mathbb{F}_{2^m}\\
					&\textup{or }  a \in \Upsilon_1, b \in \mathbb{F}_{2^m}^*, \textup{Tr}\left( ab\right) =1,\\
					2^{m-1}(2^{m-1}-1) \textup{ or } 2^{m-1}(2^{m-1}+1) ,& \textup{otherwise}.
				\end{cases}
			\end{align*}
			Given the weight $w_1=2^{m-1}(2^{m-1}-1)  $, $w_2= 2^{2m-2}$ and $w_3= 2^{m-1}(2^{m-1}+1) $, their corresponding frequencies are $A_{w_1}$, $A_{w_2}$ and $A_{w_3}$, respectively. According to Lemma \ref{Pless}, Corollary \ref{a1} and Lemma \ref{T}, their multiplicities are
			\begin{align*}
				& A_{w_1}=  2^{m-2}(2^{m-1}-1),  \\
				& A_{w_2}= 3 \cdot 2^{2m-2}-1 , \\
				& A_{w_3}=  2^{m-2}(2^{m-1}+1).  
			\end{align*}
			\begin{enumerate}[\textbf{The proof for Theorem \ref{Thm2}}]
				\item 
				
			\end{enumerate}
			
			By the same idea, the length of  $C_{D_2}$ equals to
			\begin{align*}
				n_2 &=\big|D_2\big|\\
				&=\dfrac{1}{2}\sum\limits_{z \in \mathbb{F}_2}\sum\limits_{x \in \mathbb{F}_{2^m}^*}\sum\limits_{y \in \mathbb{F}_{2^m}}(-1)^{z\textup{Tr}\left(yx^2+x+y\right)}\\
				&=(2^m-1)2^{m-1}+\dfrac{1}{2}\sum\limits_{x \in \mathbb{F}_{2^m}^*} (-1)^{\textup{Tr}(x)}\sum\limits_{y \in \mathbb{F}_{2^m}}(-1)^{\textup{Tr}\left((x^2+x)y\right)}\\
				&= (2^m-1)2^{m-1}+(-1)^{\textup{Tr}(1)}2^{m-1}\\
				&=2^m(2^{m-1}-1).
			\end{align*}
			
			Define 
			\begin{align*}
				N_2=\big|\left\{(x,y)\in \mathbb{F}_{2^m}^* \times \mathbb{F}_{2^m} :\textup{Tr}\left(yx^2+x+y\right)=0 \textup{ and } \textup{Tr}(axy+bx)=0 \right\} \big|
			\end{align*}
			then the weights of codeword $\textbf{c}$ from linear code $C_{D_2}$, 
			\begin{align*}
				wt_2(\textbf{c})&=n_2-N_2\\
				&=n_2-\dfrac{1}{2^2}\sum\limits_{z_1\in \mathbb{F}_2}\sum\limits_{z_2\in \mathbb{F}_2}\sum\limits_{x\in \mathbb{F}_{2^m}^*}\sum\limits_{y\in \mathbb{F}_{2^m}}(-1)^{z_1\textup{Tr}\left(yx^2+x+y\right)+ z_2\textup{Tr}(axy+bx)  }\\
				&=2^{m-1}(2^{m-1}-1)-\dfrac{1}{2^2}S_1-\dfrac{1}{2^2}S_3.
			\end{align*}
				By Lemma \ref{S1} and Lemma \ref{S3}, we have
				\begin{align*}
					wt_2(\textbf{c})=
					\begin{cases}
						2^{m-1}(2^{m-1}-1), & \textup{if }  a=0, b \in \mathbb{F}_{2^m}^*, \textup{Tr}(b)=1\\
						&  \textup{or } a \notin \Upsilon_1 \cup \{0\},  b \in \mathbb{F}_{2^m}\\
						&  \textup{or } a \in \Upsilon_1,  b \in \mathbb{F}_{2^m},\textup{Tr}((b+1)a)=1,\\
						2^{m}(2^{m-2}-1) \textup{ or } 2^{2m-2}, & \textup{otherwise}.
					\end{cases}
				\end{align*} 
				Given the weights $w_4= 2^{m}(2^{m-2}-1)$, $w_5= 2^{m-1}(2^{m-1}-1)$ and $w_6= 2^{2m-2}$, their corresponding frequencies are $A_{w_4}$, $A_{w_5}$ and $A_{w_6}$, respectively. 
				According to Lemma \ref{Pless} and Corollary \ref{a1} their multiplicities are
				\begin{align*}
					& A_{w_4}= 2^{m-2}(2^{m-1}-1),  \\
					& A_{w_5}= 3 \cdot 2^{2m-2} , \\
					& A_{w_6}= 2^{m-2}(2^{m-1}+1)-1 . 
				\end{align*}

				\begin{enumerate}[\textbf{The proof for Theorem \ref{thm3}}]
					\item 
				\end{enumerate}
				
				According to the defining sets \eqref{det1} and \eqref{det3}, $C_{D_1}$ and $C_{D_3}$ have the same length, that is, $n_3=n_1=2^{2m-1}$.		
				
				Define 
				\begin{align*}
					N_3=\big|\left\{(x,y)\in \mathbb{F}_{2^m}^* \times \mathbb{F}_{2^m} :\textup{Tr}\left(yx^2+xy\right)=0 \textup{ and } \textup{Tr}(axy+bx)=0 \right\} \big|,
				\end{align*}
				then the weights of codeword $\textbf{c}$ from linear code $C_{D_3}$ are
				\begin{align*}
					wt_3(\textbf{c})&=n_3-N_3\\
					&=n_3-\dfrac{1}{2^2}\sum\limits_{z_1\in \mathbb{F}_2}\sum\limits_{z_2\in \mathbb{F}_2}\sum\limits_{x\in \mathbb{F}_{2^m}^*}\sum\limits_{y\in \mathbb{F}_{2^m}}(-1)^{z_1\textup{Tr}\left(yx^2+xy\right)+ z_2\textup{Tr}(axy+bx)  }\\
					&=\dfrac{1}{2}n_3-\dfrac{1}{2^2}S_1-\dfrac{1}{2^2}S_4.
				\end{align*}
				Plugging in Lemma \ref{S1} and Lemma \ref{S4}, we get the weights that
				\begin{align*}
					wt_3(\textbf{c})=
					\begin{cases}
						2^{m-2}(2^m-1) ,      & \textup{if }  a \in \Upsilon_2,  b \in \mathbb{F}_{2^m} , \textup{Tr}((a+1)b)=0  ,\\
						2^{2m-2}         ,    & \textup{if }  a=0,  b \in \mathbb{F}_{2^m}^* , \textup{Tr}(b)=0 \\
						& \textup{or }  a=1, b \in \mathbb{F}_{2^m}  ,\\ 
						2^{m-2}(2^m+1)   ,    & \textup{if } a \in \Upsilon_2,    b \in \mathbb{F}_{2^m}^* , \textup{Tr}((a+1)b)=1   ,\\	
						2^{m-1}(2^{m-1}+1),   & \textup{if } a=0,  b \in \mathbb{F}_{2^m}^* , \textup{Tr}(b)=1  .
					\end{cases}
				\end{align*}
				The desired conclusion then follows from Lemma \ref{T}.

				It can be calculated from the Lemma \ref{Pless} that the minimum weights of the dual codes of codes we investigated are greater than or equal to $3$, that is, the three classes of codes we have constructed are projective codes. Table \ref{table4} gives some examples, according to the Griesmer bound from \cite{huffman2010}, the code $[8,4,3]$ is almost optimal.
				
				\begin{table}[H]
					\caption[short text]{Some examples }
					\centering
					\label{table4}
					\begin{tabular}{llll}
						\hline
						& m &  parameters & weight enumerater\\ \hline
						Theorem \ref{thm1} & 2 & [8,4,2] & $1+x^2+11x^4+3x^6$ \\	\hline	
						Theorem \ref{thm1} & 3 & [32,6,12] & $1+6x^{12}+47x^{16}+10x^{20}$\\	\hline
						Theorem \ref{Thm2} & 3 & [24,6,8] & $1+6x^8+48x^{12}+9x^{16}$\\	\hline	
						Theorem \ref{Thm2} & 5 & [480,10,224] & $1+120x^{224}+768x^{240}+135x^{256}$\\ \hline
						Theorem \ref{thm3} & 2 & [8,4,3] & $1+4x^3+5x^4+4x^5+2x^6$\\	\hline	
						Theorem \ref{thm3} & 3 & [32,6,14] & $1+24x^{14}+11x^{16}+24x^{18}+4x^{20}$\\	\hline
						
					\end{tabular}
				\end{table} 
				\section{Applications}
				\subsection{Applications to secret sharing schemes}
				
				In 1979, Shamir \cite{shamir1979} and Blakley \cite{gr1979} first introduced the concept of secret sharing schemes. 
				
				Secret sharing schemes can be constructed using any linear codes over $\mathbb{F}_q$ by defining an access structure. However, determining the access structure based on any linear code is often complicated and can only be done in special cases. 
				The dual codes of minimal codes are often used to construct the access structure of secret sharing schemes, so it is meaningful to search for minimal codes.
				
				For a linear code $C$, if any nonzero
				codeword $\textbf{c} \in C$ covers only its multiples, but no other nonzero codewords in $ C$, we call $C$ a minimal linear code. The following lemma provides a sufficient condition for a linear code to be minimal.
				\begin{lemma}\label{abbound}\textup{(\cite[Ashikhmin-Barg Bound]{ab1998})}
					Every nonzero codeword of a linear  code $C$ over $\mathbb{F}_p$ is minimal, provided that
					\begin{equation}
						\dfrac{w_{min}}{w_{max}}>\dfrac{p-1}{p},\notag
					\end{equation}
					where $w_{min}$ and $w_{max}$ denote the minimum and maximum nonzero weights in $C$, respectively.
				\end{lemma}
				If $m \geq 3$, the above inequality holds in Lemma \ref{abbound} for $C_{D_i}, i=1,2,3$. Hence, the linear code we constructed is minimal when $m \geq 3$. Therefore, utilizing the dual code $C_{D_i}^{\perp},i=1,2,3$ can be used to create secret sharing schemes.

				\subsection{Applications to $s$-sum sets}
				
				In 1984, Courteau and Wolfmann introduced
				the notion of the triple-sum set \cite{wolf1984}, a natural generalization of partial difference sets,
				which enter the study of two-weight codes \cite{cal1986, ma1994}. Recently, Shi and Sol\'{e} generalized
				the notion of the triple-sum set to the $s$-sum set ($s> 1$), and they gave the connection among  $s$-sum sets, $s$-strongly walk-regular graphs and three-weight codes \cite{kie2020, shi2019}.

				Some notions and results for the  $s$-sum set are given as follows \cite{shi2019}.
				Let  $ \mathbb{ F}_q$ denote the finite field of order
				$q$. The set $\Omega 
				\subsetneq \mathbb{ F}_q^k$ is a $s$-sum set if it is stable by scalar multiplication and if there
				are constants $\sigma_0$ and $\sigma_1$ such that a nonzero $\textbf{h} \in  \mathbb{ F}_q^k$ can be written as
				\begin{align*}
					\textbf{h}=\sum_{i=1}^{s}\textbf{x}_i, \textbf{x}_i \in \Omega,
				\end{align*}
				$\sigma_0$ times if $\textbf{h} \in \Omega$, and $\sigma_1$  times if $\textbf{h} \in \mathbb{ F}_q^k \setminus \Omega$. 
				If $\Omega$ and  $\textbf{0} \notin \Omega$, we denote by $C(\Omega)$ the projective code of length $n = \dfrac{|\Omega|}{q-1} $ obtained
				as the kernel of the $k \times n$  matrix $H$ with columns the projectively non-equivalent nonzero
				elements of $\Omega$. Thus $H$ is the check matrix of $C(\Omega)$.
				The following result is given in \cite{shi2019}.
				
				\begin{lemma}\label{ $s$-sum}
					Assume that $C(\Omega)^{\perp}$ is of length $n$ and has three nonzero weight $\omega_1 < \omega_2=\dfrac{n(q-1)}{q} < \omega_3$, with $\omega_1+\omega_3=\dfrac{2n(q-1)}{q}$. Then $\Omega$ is an $s$-sum set for any odd $s > 1$.
				\end{lemma}
				
				Now we can construct two $s$-sum sets  for any odd $s >1$ from $C_{D_1}$ in Theorem \ref{thm1} and $C_{D_2}$ in Theorem \ref{Thm2}.
				
				Let $\{\mu_1,\dots ,\mu_m\}$ be a basis of $\mathbb{ F}_{2^m}$ over $\mathbb{ F}_2$.
				For $n=\big|D_1\big|=2^{2m-1}$ and $D_1=\{(x_1,y_1,),\dots,(x_n,y_n)\}$, then $C_{D_1}$ has generator matrix $G=(\textbf{g}_1,\dots,\textbf{g}_n)$, where
				\begin{align*}
					\textbf{g}_i=\big(\textup{Tr}(\mu_1y_ix_i^2),\dots,\textup{Tr}(\mu_my_ix_i^2),\textup{Tr}(\mu_1y_i),\dots,\textup{Tr}(\mu_my_i)\big)^T.
				\end{align*}
				Let 
				\begin{align*}
					\Omega_{D_1}=\{\textbf{g}_i:i=1,\dots,n\} \subsetneq \mathbb{ F}_2^{2m}.
				\end{align*}
				For $m$ odd, $n=\big|D_2\big|=2^{m}(2^{m-1}-1)$ and $D_2=\{(x'_1,y'_1,),\dots,(x'_n,y'_n)\}$, then $C_{D_2}$ has generator matrix $G=(\textbf{g}'_1,\dots,\textbf{g}'_n)$, where
				\begin{align*}
					\textbf{g}'_i=\big(\textup{Tr}(\mu_1y'_i{x'}_i^2),\dots,\textup{Tr}(\mu_my'_i{x'}_i^2),\textup{Tr}(\mu_1x'_i),\dots,\textup{Tr}(\mu_mx'_i),\textup{Tr}(\mu_1y'_i),\dots,\textup{Tr}(\mu_my'_i)\big)^T.
				\end{align*}
				Let
				\begin{align*}
					\Omega_{D_2}=\{\textbf{g}'_i:i=1,\dots,n\} \subsetneq \mathbb{ F}_2^{3m}.
				\end{align*}
				Then we have the following theorems.
				\begin{theorem}
					$\Omega_{D_1}  $ is an $s$-sum set for any odd $s > 1$.
				\end{theorem}
				\begin{proof}
					$C_{D_1}$ is projective, thus $\textbf{g}_i \ne \textbf{g}_j$ for $i \ne j $, and then the weight distribution for $C(\Omega_{D_1})$ is the same as that of $C_{D_1}^{\perp}$. So that  $C(\Omega_{D_1})^{\perp}$ is a three-weight code with three nonzero weights
					\begin{align*}
						\omega_1=2^{m-1}(2^{m-1}-1), \omega_2= 2^{2m-2}, \omega_3=2^{m-1}(2^{m-1}+1),
					\end{align*}
					and then  $\omega_1 < \omega_2=\dfrac{n}{2} < \omega_3$ and $\omega_1+\omega_3=n$, from Lemma \ref{ $s$-sum},  $\Omega_{D_1} \cup \{\textbf{0}\} $ is an $s$-sum set for any odd $s > 1$.
				\end{proof}

				\begin{theorem}
					For $m$ odd, $\Omega_{D_2}  $ is an $s$-sum set for any odd $s > 1$.
				\end{theorem} 
				
				\begin{proof}
					When $m$ is odd, $C_{D_2}$ is projective, thus $\textbf{g}'_i \ne \textbf{g}'_j$ for $i \ne j $, and then the weight distribution for $C(\Omega_{D_2})$ is the same as that of $C_{D_2}^{\perp}$. So that  $C(\Omega_{D_2})^{\perp}$ is a three-weight code with three nonzero weights
					\begin{align*}
						\omega_1=2^{m}(2^{m-2}-1), \omega_2= 2^{m-1}(2^{m-1}-1), \omega_3=2^{2m-2},
					\end{align*}
					and then  $\omega_1 < \omega_2=\dfrac{n}{2} < \omega_3$ and $\omega_1+\omega_3=n$, from Lemma \ref{ $s$-sum},  $\Omega_{D_1} \cup \{\textbf{0}\} $ is an $s$-sum set for any odd $s > 1$.
				\end{proof}


				\section{Conclusion}
				In this paper, three new classes of projective three-weight and
				four-weight linear codes over $\mathbb{ F}_2$ were constructed. Then  $s$-sum sets for any odd $s>1$ were constructed by using the projective three-weight codes proposed in this paper.
				As a further work, we will find more few-weight codes by using different defining sets. Also we will attempt to construct new classes of  minimal codes, which have applications in the secret sharing by using the similar method in this paper and the  projective three-weight codes can be used to construct $s$-sum sets for any odd $s > 1$.

				\section*{Acknowledgements}
				
				The authors would like to thank the reviewers and editors for their
				detailed and constructive comments, which substantially improved the
				presentation of the paper. This work was supported by National Natural
				Science Foundation of China (No. 62272420), Natural Science Foundation of Fujian Province (No. 2023J01535).

		\end{document}